%% file: main.tex
\documentclass[sigconf]{acmart}

\input{packages}

\AtBeginDocument{%
  }

\setcopyright{acmlicensed}
\copyrightyear{2024}
\acmYear{2024}
\acmDOI{XXXXXXX.XXXXXXX}

\acmISBN{978-1-4503-XXXX-X/18/06}

\begin{document}

\title{GraphWeaver: Billion-Scale Cybersecurity Incident Correlation 
}

\author{Scott Freitas}
\email{scottfreitas@microsoft.com}
\orcid{XXXX}
\affiliation{%
  \institution{Microsoft Security Research}
  \city{Redmond}
  \state{Washington}
  \country{USA}
}

\author{Amir Gharib}
\email{agharib@microsoft.com}
\orcid{XXXX}
\affiliation{%
  \institution{Microsoft Security Research}
  \city{Redmond}
  \state{Washington}
  \country{USA}
}

\renewcommand{\shortauthors}{Freitas \& Gharib}

\input{sections/00-abstract}



\received{20 May 2024}

\maketitle

\input{sections/01-introduction}

\input{sections/02-related}

\input{sections/03-problem}
\input{sections/04-architecture}
\input{sections/05-deployment}

\input{sections/06-research-directions}
\input{sections/07-conclusion}
\input{sections/08-acknowledgements}

\bibliographystyle{ACM-Reference-Format}
\bibliography{main}

\end{document}

%% file: packages.tex
\usepackage{bm}
\usepackage{enumitem}
\DeclareMathAlphabet{\altmathcal}{OMS}{cmsy}{m}{n}
\usepackage[ruled, noend]{algorithm2e}
\usepackage{graphicx}
\usepackage{textcomp}
\usepackage{xcolor}
\usepackage{lipsum}
\newtheorem{problem}{Problem}
\newtheorem{lemma}{Lemma}

\newcommand{\method}{\textsc{GraphWeaver}}

\newcommand{\hide}[1]{}

%% file: sections/00-abstract.tex
\begin{abstract}
In the dynamic landscape of large enterprise cybersecurity, accurately and efficiently correlating billions of security alerts into comprehensive incidents is a substantial challenge. 
Traditional correlation techniques often struggle with maintenance, scaling, and adapting to emerging threats and novel sources of telemetry.
We introduce \method{}, an industry-scale framework 
that shifts the traditional incident correlation process to a data-optimized, geo-distributed graph based approach.
\method{} introduces a suite of innovations tailored to handle the complexities of correlating billions of shared evidence alerts across hundreds of thousands of enterprises. 
Key among these innovations are a geo-distributed database and PySpark analytics engine for large-scale data processing, a minimum spanning tree algorithm to optimize correlation storage, integration of security domain knowledge and threat intelligence, and a human-in-the-loop feedback system to continuously refine key correlation processes and parameters.
\method{} is integrated into the Microsoft Defender XDR product and deployed worldwide, handling billions of correlations with a 99\% accuracy rate, as confirmed by customer feedback and extensive investigations by security experts.
This integration has not only maintained high correlation accuracy but reduces traditional correlation storage requirements by 7.4x. 
We provide an in-depth overview of the key design and operational features of \method{}, setting a precedent as the first cybersecurity company to openly discuss these critical capabilities at this level of depth.
\end{abstract}

%% file: sections/01-introduction.tex
\section{Introduction}
The exponential growth of threat actors, coupled with the proliferation of cybersecurity solutions aimed at thwarting them, has inundated security operation centers (SOCs) with a flood of alerts~\cite{forrester2020state}. 
Amidst this deluge, discerning genuine threats from the noise presents a formidable challenge. 
In response, alert correlation (Figure~\ref{fig:crown}) has become an indispensable tool in the defender's arsenal, allowing SOCs to consolidate disparate alerts into cohesive incidents, dramatically reducing the number of analyst investigations~\cite{ban2023breaking}.

Intuitively, correlation can be likened to the process of ``weaving'' together security alerts into cohesive incident narratives, grounded in shared indicators of compromise such as malicious files or IP addresses. 
To facilitate this process, Extended Detection and Response (XDR) platforms, such as Microsoft Defender XDR, take on the pivotal role of centralized telemetry hub for the myriad of security products used by organizations. 
One of the primary goals of these XDR platforms is to enhance SOC efficiency and effectiveness by synthesizing disparate alerts across endpoint, identity, email, collaboration tools, cloud services, and data repositories, into cohesive incident graphs that serve as a representation of threat activity occurring in the enterprise.

\begin{figure*}[t!]
    \centering
    \includegraphics[width=\textwidth]{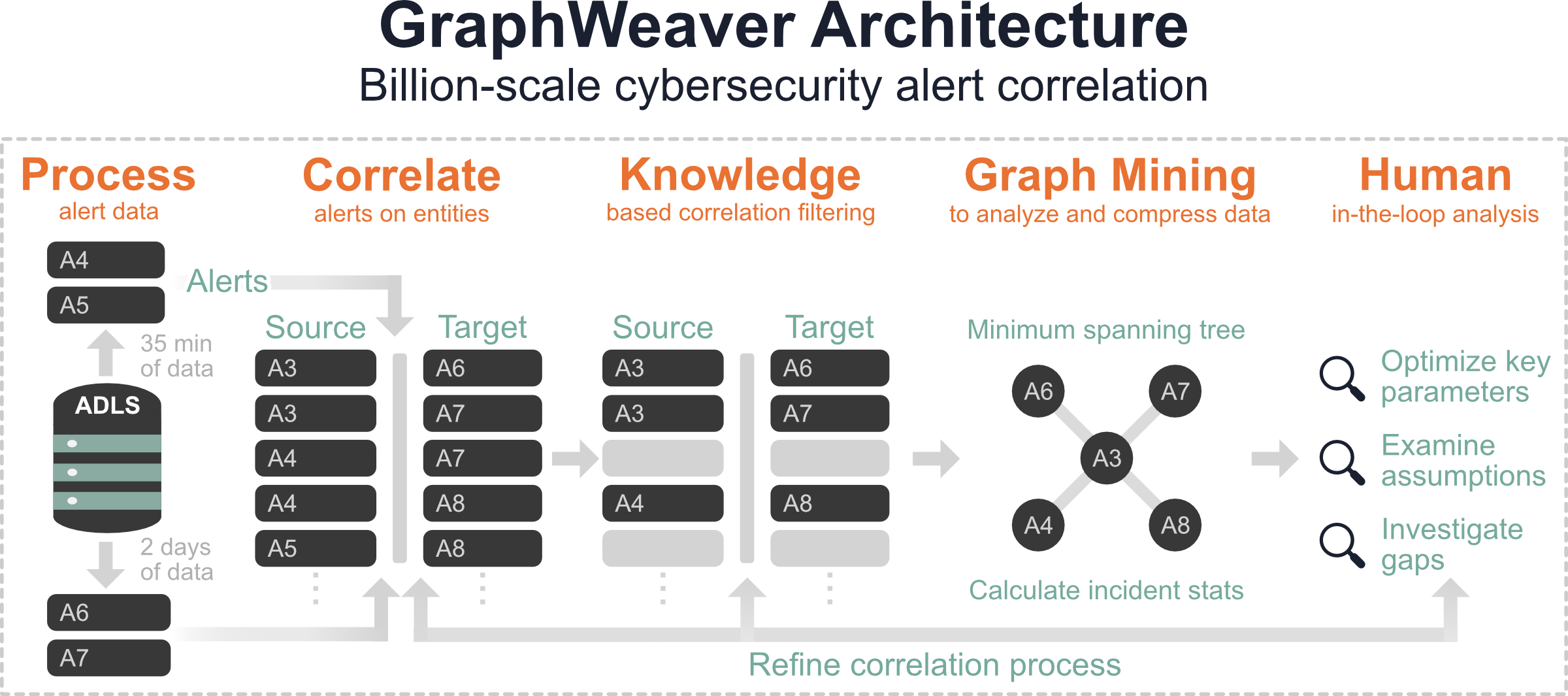}
    \caption{Overview of the \method{} architecture from the perspective of a single geographic region.
    Alert telemetry is retrieved from an ADLS system and processed into two PySpark dataframes: one containing the latest 35 minutes of data, and the other the last 2 days.
    These dataframes are then correlated based on shared entities between pairs of alerts.
    Subsequently, threat intelligence and security domain knowledge are applied to filter out invalid correlations.
    The refined correlations are then centralized to the main node and converted into a graph, enabling a minimum spanning tree algorithm to remove redundant incident correlations.
    Finally, a human-in-the-loop feedback system is employed, where threat researchers review correlation reports to optimize parameters, review assumptions, and pinpoint potential correlation gaps.
    }
    \label{fig:crown}
\end{figure*}

\vspace{1mm}\noindent
\textbf{Incident correlation at scale.}
Scalable and efficient incident correlation presents multiple unique and exciting challenges:

\begin{enumerate}[ topsep=4pt, leftmargin=*, itemsep=3pt]
    \item \textbf{Mitigating false correlations.} False correlations pose a significant risk, potentially leading to unwarranted actions on benign devices or users, disrupting vital company operations. 
    Additionally, over-correlation can result in ``black hole'' incidents, where all alerts within an enterprise begin to correlate indiscriminately.
    
    \item \textbf{Minimizing missed correlations.} Avoiding false negatives is equally important, as a missed correlation could be the difference between the key context required to disrupt a cyberattack, preventing the loss of valuable data and intellectual property.
    
    \item \textbf{Scalability across vast telemetry.} Correlating billions of alerts across a multitude of security products presents a monumental scaling challenge, requiring a robust infrastructure and an efficient methodology.
    Furthermore, these correlations need to happen in near real-time to keep SOCs up to date.
    
    \item \textbf{TI \& domain knowledge.} Correlation across diverse entity types such as IP addresses and files often requires specialized threat intelligence (TI) and domain knowledge to mitigate false positive and false negative correlations.
    
\end{enumerate}

\noindent
The emergence of XDR as a relatively new industry field underscores the timeliness of these challenges, and positions the realm of alert correlation as a pivotal frontier in the field of cybersecurity. 
Going forward, innovative solutions will be imperative to successfully navigate this intricate terrain and safeguard organizations against ever-evolving threats.

\subsection{Contributions}
We introduce \method{} (Fig.~\ref{fig:crown}), a novel framework designed to tackle the challenging task of correlating security incidents at scale, marking the first in-depth academic-industry discourse on this vital endeavor. 
Our framework makes significant contributions to the following areas:

\begin{itemize}[topsep=2mm, itemsep=0mm, parsep=1mm, leftmargin=*]
    
    \item \textbf{\method{}.} 
    The \method{} architecture reshapes cybersecurity incident correlation by introducing a scalable framework capable of correlating alerts at the billion-scale. 
    Integrating domain knowledge and threat intelligence, \method{} ensures entity contextualized correlations to minimize both false and missing correlations.   
    Its efficiency is enhanced by a minimum spanning tree algorithm that compresses the number of correlations required to complete an incident by 7.4x compared to traditional rule-based approaches.
    Additionally, \method{} constantly adapts by mining filtered incident patterns, refining results through a human-in-the-loop feedback system. 
    Most importantly, we reveal key architectural design elements and operational processes, setting a precedent as the first cybersecurity company to openly discuss correlation at this level of depth.

    \item \textbf{Impact to Microsoft Customers and Beyond.} 
    \method{} is integrated into the Microsoft Defender XDR product, a leader in the market~\cite{mellen2024forrester}, which is deployed to hundreds of thousands of organizations worldwide.
    In the realm of XDR, accurately and efficiently correlating alerts into comprehensive incidents stands as one of the most formidable challenges. 
    This research has led to major impact, reshaping the products approach to correlation at scale.

    \item \textbf{Open Research Challenges for the Community.}
    We discuss a spectrum of open research directions that have the potential to substantially impact the cybersecurity correlation landscape. 
    Among these lies the automatic integration of novel sources of threat intelligence and domain knowledge, leveraging large language models (LLMs) and knowledge graphs to enhance correlation efficacy. 
    Additionally, bridging correlation gaps by discerning how indirectly related alerts and incidents, devoid of shared evidence, can correlate through alternative processes like representation learning presents is a key challenge. 
    Moreover, the pursuit of verifying correlation accuracy at scale, whether through LLMs or alternative methodologies, holds the promise of unlocking new insights and mitigating the impact of false correlations. 
    Together, these multifaceted challenges not only push the boundaries of innovative research but contribute to a more secure future.
\end{itemize}

To enhance readability, Table~\ref{table:terminology} details the correlation terminology used in this paper. 
The reader may want to refer to this table for meanings and synonyms of technical terms.

%% file: sections/02-related.tex
\input{tables/terminology}

\section{Related Work}\label{sec:Related}
Our work is informed by advancements in key related fields: Network Intrusion Detection Systems (NIDS), Security Information and Event Management (SIEM), 
and Extended Detection and Response (XDR). 
Together, these domains form a symbiotic relationship, each building upon the other to enhance the efficacy and scope of incident correlation.
We select a number of academic and industry works relevant to each domain. 
However, before the introduction of \method{}, industry-based solutions were not publicly discussed in depth, hindering an open discussion of research in the correlation field.

\subsection{Network Intrusion Detection Systems}
NIDS specialize in monitoring and analyzing network traffic for indications of malicious activity~\cite{roesch1999snort,kr2010intrusion,freitas2020d2m}. 
These systems operate by inspecting network packets in real-time, employing signature and anomaly detection techniques to discern threats from suspicious behavior. 
Signature-based detection entails matching network traffic against a database of known attack signatures, identifying threats like malware payloads and network exploits~\cite{masdari2020survey}. 
Conversely, anomaly-based detection identifies statistical deviations from normal network behavior, such as unusual traffic patterns or protocol violations, signaling potential security breaches~\cite{gamage2020deep,chou2021survey}.
Following the identification of network alerts and events, NIDS systems correlate relevant telemetry into incidents using a variety of approaches encompassing heuristic~\cite{wu2019alert,elshoush2013intrusion,dain2002fusing,bateni2013using}, rule-based~\cite{kabiri2007rule,alserhani2016alert,cuppens2001managing}, and machine learning techniques~\cite{cheng2021discovering,smith2008using,ramaki2015real,xuewei2014approach,haas2019alert} 
In addition, a variety of surveys can be found on NIDS-based correlation methodologies~\cite{yu2014survey,salah2013model,sadoddin2006alert,mirheidari2013alert}.
This correlation task is relatively straightforward, focusing on a reduced subset of network-centric alerts rather than across the entire enterprise threat landscape.

\subsection{Security Information \& Event Management}
SIEM systems play a pivotal role as centralized platforms for gathering, analyzing, and correlating security telemetry sourced from various platforms within an organization's IT infrastructure~\cite{gartner2024siem,vielberth2020security}.
These solutions collect logs, events, and alerts generated by a myriad of sources including network devices, servers, applications, and security control systems such as firewalls, antivirus software, and intrusion detection~\cite{gonzalez2021security}. 
However, SIEMs are best understood as platforms designed for customization and flexibility, allowing each organization to tailor the environment to suit their unique preferences and requirements. 
While there is limited research into SIEM based alert correlation~\cite{granadillo2016new,kotenko2022systematic}, these environments heavily rely on the expertise and domain knowledge of the organization's security team to craft correlation rules, rather than out-of-the-box support provided by the system itself.

\subsection{Extended Detection and Response}
XDR solutions represent a significant leap forward in the realm of alert correlation, particularly when compared to other cybersecurity fields. 
They offer easy integration with various security products such as network, email, endpoint, and cloud security, enabling visibility and correlation capabilities across an organization's entire digital footprint.
However, cross-domain correlation is a formidable challenge, as XDR systems must efficiently and accurately correlate across diverse enterprise landscapes. 
                                                                               
There are several prominent XDR products that provide alert correlation by applying static rules based on shared entities like User, IP address, and URL \cite{CiscoXDRIncident, paloaltoXDRincident}. 
However, the transparency regarding the methodologies and data used by these companies is limited, with only minimal information made publicly available.
This lack of detailed public disclosure poses a significant challenge for academic research and development in the field. 
We believe that our research will not only advance these technological capabilities but foster greater openness and collaboration within the cybersecurity research community.

%% file: tables/terminology.tex
\begin{table*}[th]
\centering
\begin{tabular}{p{0.2\textwidth} p{0.235\textwidth} p{0.5\textwidth}}

\textbf{Technical term} & \textbf{Synonyms} & \textbf{Meaning} \\
\cmidrule(r){1-1} \cmidrule(lr){2-2} \cmidrule(l){3-3}
Alert & Node, vertex & Potential security threat that was detected on key related entities  \\ \addlinespace
Correlation & Edge, link, connection & Connection between two alerts based on a shared entity \\ \addlinespace
Enterprise & Organization, company & Organizations containing an XDR product \\ \addlinespace
Detector & Rule & A security rule, algorithm, or ML model that generates alerts. 
Detectors can be created by SOCs (custom) or provided by default (built-in) \\ \addlinespace
Incident & Graph, subgraph  & Correlated alerts from email, endpoint, cloud, and network layers that reveal comprehensive threat actor activities \\ \addlinespace
Entity & Evidence & An entity is file, IP address, etc. associated with an alert \\ \addlinespace
Time window & Time span & The time span between two alerts  \\ \addlinespace
Max time window & Entity time window & Maximum allowed time between two alerts sharing a particular entity  \\ \addlinespace
Entity Priority & Priority score & Priority score assigned to each edge based on associated entity type \\ \addlinespace
True positive correlation & TP link & Correlation between two alerts is correct \\ \addlinespace
False positive correlation & FP link & Correlation between two alerts is incorrect \\ \addlinespace
True negative correlation & TN link & Correlation between two alerts is correctly ignored (non-existent) \\ \addlinespace
False negative correlation & FN link & Correlation between two alerts is incorrectly ignored (non-existent) \\ \addlinespace
\bottomrule
\end{tabular}
\caption{Incident correlation terminology containing meanings and synonyms of technical terms.}
\vspace{-5mm}
\label{table:terminology}
\end{table*}

%% file: sections/03-problem.tex
\section{Problem Formulation}\label{sec:problems}
A core business challenge presented by XDR is the need to efficiently and accurately correlate alerts across a vast digital infrastructure, spanning hundreds of thousands, or even millions of enterprises. 
This challenge demands a solution that can handle complex data interactions while adhering to stringent operational and regulatory requirements, including:

\begin{itemize}[topsep=4pt, leftmargin=*, itemsep=3pt]
    \item \textbf{Geo-distributed computation and storage.}
    To comply with international data privacy regulations, the system must support geo-distributed computation and storage to ensure data residency and sovereignty requirements.

    \item \textbf{High-frequency batch processing.}
    The system must continuously operate in streaming or batch mode, processing alerts every few minutes. 
    This ensures that alerts are promptly correlated, enabling quick response to security incidents.

    \item \textbf{Scalability for high-volume data.}
    The architecture must be capable of scaling to accommodate tens of millions of alerts per region in each run without degradation in performance. 

    \item \textbf{Low tolerance for incorrect correlations.}
    The system needs to incorporate security domain knowledge and threat intelligence to minimize false positive and false negative correlations.

    \item \textbf{System redundancy and robustness.}
    To ensure continuous operation and mitigate customer outages, robust redundancy mechanisms are needed to handle potential job failures without dropping alert correlations.

    \item \textbf{Parameter optimization and pattern mining.}
    Optimal parameter settings for correlation time windows must be determined for each entity type to maximize the effectiveness of alert correlations. 
    Additionally, the system should employ pattern mining techniques on the correlations to identify and raise potential incident gaps for resolution. 
\end{itemize}

\noindent
By viewing these requirements through the lens of efficient graph correlation, we formulate two research tasks: 
(1) scalable incident graph mining from a dataframe of raw alert telemetry;
and (2) parameter optimization and correlation gap discovery by mining potentially missed correlations. 
We formally define each problem below and present our solutions in Section~\ref{sec:methodology}.

\begin{problem}\label{problem:1}\textbf{Scalable Incident Graph Mining}

\begin{description}[topsep=1mm, itemsep=0mm, parsep=1mm, leftmargin=6mm, itemindent=0mm]

\item [Given.] A dataframe \textbf{T} containing alert telemetry.

\item [Find.] A correlation dataframe $\bm{C}$ satisfying six constraints: 
    \begin{enumerate}
    \item Each correlation is \textit{valid} based on the time difference between two alerts and the shared entity's maximum time period.

    \item Each correlation is \textit{valid} based on the latest threat intelligence. 

    \item Cross-detector correlation pass key safety checks.
    
    \item Multiple correlations between alerts are resolved such that only the highest prioritized one is retained.
    
    \item Only the minimum set of correlations needed to create an incident are retained.
    
    \end{enumerate}
\end{description}
\end{problem}

\smallskip

\begin{problem}\label{problem:2}\textbf{Parameter Optimization \& Gap Discovery}

\begin{description}[topsep=1mm, itemsep=0mm, parsep=1mm, leftmargin=6mm, itemindent=0mm]

    \item [Given.] A dataframe \textbf{T} containing alert telemetry.
    
    \item [Find.] Optimized correlation time windows for each entity type, and potential incident correlation gaps.
\end{description}
\end{problem}

\medskip\noindent
\textbf{Assumptions.} 
We establish the following assumptions during our analysis.
(1) \textit{Alert classification independence}, where individual alert grades (e.g., TP, FP) are not considered in the correlation process.
We observe that a well-constructed correlation architecture inherently associates TP and FP alerts.
(2) \textit{Homogeneous graph model} consisting exclusively of alert nodes. 
Entities are included as edge metadata rather than as separate nodes. 
Although entities can be modeled as nodes to facilitate indirect alert correlations, our findings suggest that direct alert correlation is more efficient.
(3) \textit{Non-directional edges} independent of alert generation time, as the focus is only on establishing the existence of relationships between alerts.

%% file: sections/04-architecture.tex
\section{\method{} Architecture}\label{sec:methodology}
\method{}'s architecture and critical design choices are detailed in eleven steps (S1-S11) across three main sections.
In Section~\ref{subsec:correlation}, we present our methodology for alert correlation at-scale using Azure Data Lake Storage Gen2 (ADLS) and an Azure Synapse PySpark engine.
Section~\ref{subsec:mining-incident-subgraphs} introduces the process of creating incident graphs from alert correlations.
Finally, Section~\ref{subsec:parameter-optimization} discusses our strategy for mining correlation patterns, allowing us to optimize key \method{} parameters and close potential correlation gaps.
While our analysis concentrates on the architecture from the perspective of a single region, it is important to note that the architecture and methodology presented can be uniformly deployed to any number of regions.
See Algorithm~\ref{alg:graphweaver} for an overview of the entire correlation process.

\subsection{Scalable Alert Correlation}\label{subsec:correlation}
We establish the foundation for addressing the first four constraints of Problem~\ref{problem:1} by processing the raw alert telemetry through a six-step process:
(S1) collect alert telemetry from the ADLS system,
(S2) correlate alerts, 
(S3) filter invalid correlations,
(S4) remove low fidelity correlations by leveraging threat intelligence,
(S5) exclude correlations that could potentially cause massive incidents (``black holes''),
and (S6) remove redundant correlations, keeping ones with higher entity prioritization. 
All correlation parameters presented in this section were established through extensive collaboration with threat researchers at Microsoft, aimed at minimizing the risk of false correlations.

\medskip\noindent
\textbf{S1---Collect telemetry.}
We collect 72 hours of historical alert telemetry into a PySpark dataframe called \textit{target alerts} $\bm{T}$. 
This collection period is strategically chosen based on the maximum entity correlation time window specified in Table~\ref{table:entity-info}.
Next, the most recent 35 minutes of telemetry is copied from the target alerts dataframe into a second dataframe called \textit{source alerts} $\bm{S}$.
Each row in these dataframes includes columns for unique organizational and alert identifiers, along with the 17 entities listed in Table~\ref{table:entity-info}.

\medskip\noindent
\textbf{S2---Correlate alerts.}
We transform the source and target alert dataframes into correlations by iteratively joining them across all 17 entity types. 
Three primary constraints guide this process---(i) alerts must share an entity to correlate, (ii) originate from the same organization, and (iii) no self-correlations. 
This approach enables us to generate correlations for every pair of alerts that share common entities, allowing multiple correlations between alerts with several shared entity types. 
Next, the 17 individual dataframes generated from these joins are merged into a single dataframe $\bm{C}$ containing all potential alert correlations. 
Each row in this unified dataframe includes a source alert, target alert, organization identifier, and columns for all 17 entities.
We then refine the correlations by merging it with another ADLS table containing existing correlations, allowing us to remove redundancies.

\medskip\noindent
\textbf{S3---Filter invalid correlations.} 
We filter the correlation dataframe to remove invalid links based on the elapsed time between alerts and the maximum allowed by entity time windows (see Table~\ref{table:entity-info}).
Specifically, a correlation between two alerts $u$ and $v$ is valid if the time difference $\Delta t$ between them does not exceed the max correlation time $t_m$ of the shared entity, which can be denoted as $|\Delta c_{uv}| \leq t_m$.
This ensures adherence to the first constraint of Problem~\ref{problem:1}
Maximum time windows for each entity were initially set based on threat researcher domain knowledge, and subsequently optimized using the pattern mining approach described in Section~\ref{subsec:parameter-optimization}.
The final optimized values are reported in Table~\ref{table:entity-info}.

\input{tables/entity-info}

\medskip\noindent
\textbf{S4---Integrate threat intelligence.} 
We enhance our correlation process by integrating threat intelligence (TI) to more accurately identify valid correlations for specific entity types such as SHA1, FileName, and IPRange. 
Given the lower fidelity nature of IPRange data, due to frequent updates and the complexity in determining whether an IPRange indicates malicious activity, a VPN setup, or simply abnormal behavior, it is crucial to incorporate TI. 
To prevent low quality correlations on benign entities, we use Microsoft Defender Threat Intelligence to keep only those correlations involving SHA1 hashes and FileNames linked to known malicious signatures. 
Likewise, correlations for IPRanges are only kept if they have been confirmed as malicious within the last 48 hours.

\medskip\noindent
\textbf{S5---Prevent black hole correlations.} 
A critical capability is safely correlating alerts generated by built-in XDR detectors, with custom detectors developed by SOC analysts.
This is especially important since nearly half of all alerts originate from custom detectors. 
While this may seem conceptually straightforward, it demands careful analysis to ensure safe correlation across the boundary of built-in and custom alerts. 
Built-in rules are known for their high fidelity detection requirements, whereas custom detections often lack standardized quality controls. 
For instance, it is not uncommon for custom detections by an enterprise to yield thousands of daily alerts---overwhelmingly noisy and of low value to SOC analysts tasked with their review.

Allowing correlations between high-volume, low-fidelity custom detectors and high-fidelity built-in detectors can lead to the formation of unmanageable ``black hole'' incidents that obscure essential alerts from the built-in detectors. 
Conversely, requiring SOC analysts to examine thousands of individual noisy alerts is not practical. 
Moreover, completely avoiding correlations between custom and built-in rules could overlook critical information pertinent to an incident. 
Therefore, it is essential to strategically enable cross detector correlations between high-fidelity custom and built-in detectors, while also allowing for the correlation of noisy alerts originating from the same rule.
To manage this delicate balance, we developed three essential safety checks to enable cross-detector correlation:

\begin{enumerate}[topsep=4pt, leftmargin=*, itemsep=3pt]
    \item \textbf{Low volume detector.} We examine the historical alert volume for each detector. 
    Cross-detector correlation is activated only if the detector’s average alert volume ($l_d$) is below 6\% of the total alerts, and fewer than 20 alerts per day across the enterprise.
    A detector generating more than 20 alerts per day across an enterprise is considered low fidelity, or noisy.

    \item \textbf{Low evidence detector.} The average number of distinct values per entity type ($l_{avg}$) in a detector should not exceed predetermined thresholds. 
    For most entity types, this threshold is set at 4. However, for entity types such as SHA1, FileName, Url, EmailId, and AppId, a higher threshold of 10 is allowed. 
    These limits are set based on threat researcher expertise help prevent ``black hole'' incidents that indiscriminately correlate unrelated alerts.
    
    \item \textbf{Low evidence alert.} Similarly, the number of distinct entity type values ($l_a$) associated with a single alert must be constrained to the same thresholds set above. 
    Since organizations can attach an arbitrary number and type of entities to detector outputs, alerts may occasionally contain dozens, if not hundreds, of entities.
\end{enumerate}

In practice, we develop a secondary PySpark job that profiles all custom and built-in detectors on an hourly basis, storing the resulting telemetry in an ADLS table to confirm their suitability for cross-detector correlation. 
This setup enables the \method{} job to dynamically access these detector profiles and make correlation decisions in real time.

\medskip\noindent
\textbf{S6---Prioritize duplicate correlations.}
At this stage of the process, numerous duplicate correlations still exist between alerts with multiple shared entities. 
The final correlation step involves selecting the most critical link between alerts based on their entity prioritization score (see Table~\ref{table:entity-info}), and filtering out the remaining ones.
The criteria for correlation link prioritization is established through partnership with threat researchers at Microsoft.

\subsection{Incident Graph Mining}\label{subsec:mining-incident-subgraphs}
We transform the correlations from the previous section into incident graphs, allowing us to address the final constraint of Problem~\ref{problem:1} through a three-step process: 
(S7) transform the correlation dataframe into millions of incident graphs,
(S8) remove redundant correlation links through a minimum spanning tree (MST) algorithm,
and (S9) mine incident graph statistics to monitor the efficacy and efficiency of the correlation process.

\medskip\noindent
\textbf{S7---Graph construction.}
After completing the computationally demanding distributed PySpark calculations, we transfer the resulting alert correlations dataframe into a Pandas dataframe for additional processing on the central PySpark node. 
This step of data conversion and centralization facilitates the construction of a simple\footnote{defined as having only one edge between any two alerts and no self-correlating edges} graph in NetworkX. 
Each node in the graph represents an alert, containing metadata such as the time of occurrence and an organization identifier, with edges characterized by the correlation entity type.

\medskip\noindent
\textbf{S8---Compress graph links.}
Our incident graph, now containing millions of subgraphs, is optimized to retain only the essential edges necessary to connect the alerts within each incident. 
For example, in an incident involving alerts $A$, $B$, and $C$ with the correlations $A \rightarrow B$, $B \rightarrow C$, and $A \rightarrow C$, we can eliminate redundant correlations like $A \rightarrow C$ or $B \rightarrow C$ to streamline the graph. 
To accomplish this compression, we run a minimum spanning tree (MST) algorithm across the entire incident graph. 
This approach ensures that we retain only the theoretical minimum number of edges required to connect each incident subgraph. 
Resolving this task addresses the final constraint of Problem~\ref{problem:1}, enhancing our ability to efficiently store correlations, and reduce the computational costs of downstream processes utilizing the telemetry.

\medskip\noindent
\textbf{S9---Mine graph statistics.}
With our incident graphs now complete, we can extract detailed statistics that provide insights into various aspects of the correlation process.
These statistics can include the number of correlations per entity type, correlations segmented by region, correlations categorized by product and detector type, the distribution of incident sizes, the average runtime of correlation processes per region, and the success and failure rates of correlation jobs.
Collecting and analyzing these and other statistics serves multiple purposes. 
First, they offer a comprehensive view of the operational health of our correlation jobs by highlighting potential bottlenecks. 
Additionally, these metrics enable targeted monitoring, allowing us to identify trends, anomalies, and potential areas requiring intervention or optimization.

\input{algorithms/graphweaver-algorithm}

\subsection{Parameter Optimization \& Gap Discovery}\label{subsec:parameter-optimization}
We address Problem~\ref{problem:2} through two steps:
(S10) optimizes correlation time windows by analyzing both valid and rejected correlations, 
and (S11) pinpoints potential correlation gaps through analysis of filtered correlations and engaging in continuous feedback with our threat research team.
This human-in-the-loop feedback system ensures that \method's correlation strategies are not only precise, but robust against evolving security challenges.

\medskip\noindent
\textbf{S10---Time window optimization.}
We optimize the correlation time window for each entity in Table~\ref{table:entity-info} by identifying potential correlation gaps. 
This process begins by analyzing both valid and filtered correlations from the output of S2 in Section~\ref{subsec:correlation}. 
We calculate key statistical measures such as the average, median, and percentiles for the correlation times of valid and invalid correlations, as well as their combined correlations, on a per-entity basis.
These statistics are then forwarded to our threat research team for a detailed investigation.
This enables them to assess whether an extension of the correlation time window for specific entities could reduce false negatives, or conversely, if a reduction is necessary to decrease false positives. 
By continuously refining these time windows based on empirical data and expert insights, we are able to optimize \method's correlation process.

\medskip\noindent
\textbf{S11---Correlation gap discovery.} 
Unoptimized time windows are not the only contributor to correlation gaps. 
Gaps can also arise from missing threat intelligence, or shifting telemetry as new products and detectors challenge existing assumptions.
To address this, we analyze rejected correlations from the processes outlined in Section~\ref{subsec:correlation} (S3-S5) to identify the most prevalent potential correlation gaps across different detectors and entity types.
The findings are then forwarded to our threat research team, allowing them to assess the need for new TI feeds, revising our correlation assumptions, or adjusting various correlation parameters to maintain and enhance the accuracy and relevance of the system.

\subsection{Computational Complexity}

\begin{lemma}
The time and space complexity of \method{} is $O(s + t + c \log c)$ and $O(sm + tm + cm)$, respectively.
\end{lemma}

\begin{proof}
The time complexity of \method{} is dominated by two key operations: the iterative joins in step S2 and the minimum spanning tree (MST) algorithm in step S8. 
Each PySpark join operation has a time complexity of $O(s + t)$, where $s$ and $t$ represent the number of rows in the source $\bm{S}$ and target $\bm{T}$ dataframes, respectively. 
Despite the join being executed for each of the 17 entities, the aggregate time complexity for the joins remains $O(s + t)$. 
The MST algorithm can be implemented with a time complexity of $O(c \log c)$, where $c$ is the number of edges in the graph---which corresponds to the number of rows in the correlation dataframe $\bm{C}$. 
When combining these complexities, the total time complexity of the method is $O(s + t + c \log c)$.

The space complexity of \method{} is determined by the memory requirements of the three primary PySpark dataframes, $\bm{S}$, $\bm{T}$, and $\bm{C}$. 
These dataframes are stored in dense format, and occupy space proportional to $O(sm)$, $O(tm)$, and $O(cm)$ respectively, where $s$, $t$, and $c$ indicate the number of rows and $m$ represents the number of columns in each dataframe. 
As a result, the overall space complexity is $O(sm + tm + cm)$.
\end{proof}

%% file: tables/entity-info.tex
\begin{table}[t]
\centering
 \begin{tabular}{l l r r} 
 \toprule
 \textbf{Entity} & \textbf{Description} & \textbf{Priority} & \textbf{Time} \\
 \midrule
 SessionId & Cloud session id & (high) 1 & 48h \\
 EmailId & Email message id & 2 & 48h \\
 CampaignId & Email campaign id & 3 & 72h \\
 EmailCluster & Email cluster id & 4 & 72h \\
 UserId & User account id & 5 & 24h \\
 URL & Website URL and domain & 6 & 48h \\
 DeviceId & Identifier for device & 7 & 24h \\
 SHA1 & Cryptographic file hash & 8 & 24h \\
 FileName &  Name of a file & 9 & 24h \\
 AppId & Identifier for cloud app & 10 & 48h \\
 EmailAddress & Email sender address & 11 & 12h \\
 EmailSubject & Email subject & 12 & 12h \\
 RegistryKey & OS registry key & 14 & 24h \\
 RegistryValue & Data stored in key  & 13 & 24h \\
 ResourceId & Cloud resource id & 15 & 24h \\
 IP & IP address & 16 & 8h \\
 IPRange & IP addresses in subnet /24 & (low) 17 & 8h \\
 
\bottomrule
\end{tabular}
\caption{This table contains a description, priority score, and maximum correlation time window for 17 entity types.
The priority scores and correlation time windows for each entity are determined in conjunction with domain experts.
}
\vspace*{-5mm}
\label{table:entity-info}
\end{table}

%% file: algorithms/graphweaver-algorithm.tex
\begin{algorithm}[!t]
\KwIn{Alert dataframe $\bm{T}$, TI entity list $N$; functions: TI lookup $L$, entity prioritization $S$, alert volume $V$, avg entity volume $A$, number of entities $E$,  detector volume $V$; and threshold limits: $t_m$, $l_d$, $l_{avg}$, and $l_a$}
\KwOut{A refined correlation dataframe $C$} 
    \BlankLine
    let $\bm{R} = \emptyset$ \tcp*{initialize temp dataframe} 
    let $\bm{S} = \{ s \in \bm{T} \mid \text{now()} - s_{time} \leq 35 \text{ minutes} \}$ \tcp*{S1} 
    \BlankLine
    \For{entity $e \in E$} {
        $\bm{R} = \bm{R} \cup \bm{S}.\text{join}(\bm{T}, \text{on} = e, \text{how} = \text{inner})$ \tcp*{S2} 
    } 
    \BlankLine
    $\bm{C} = \{c \in \bm{R} \mid |c_{time_1} - c_{time_2}| \leq t_m \}$ \tcp*{S3} 
    $\bm{C} = \{c \in \bm{C} \mid c_e \in N \wedge L(c) = \text{malicious}\}$ \tcp*{S4} 
    $\bm{C} = \{c \in \bm{C} \mid V(c) \leq l_d \wedge A(c) \leq l_{avg} \wedge E(c)\leq l_a \}$ \tcp*{S5} 
    $\bm{C}$ = \text{$\bm{C}$[$\bm{C}$.GroupBy($\bm{C}_{src}$, $\bm{C}_{tar}$)[$\bm{C}_{score}$].argmin()]} \tcp*{S6} 
    $\bm{C} = \text{DataFrame}(\text{MST}(\text{Graph}(\bm{C})))$ \tcp*{S7, S8} 
    \BlankLine
    $\bm{F} = \bm{R} \setminus \bm{C}$ \tcp*{filtered correlations}
    $t_m \gets \text{Optimize}(\text{Feedback}(\text{Stats}(\bm{C}, \bm{F})))$ \tcp*{S9, S10} 
    $\text{Assumptions} \gets \text{Feedback}(\text{Analyze}(\bm{C}, \bm{F}))$ \tcp*{S11} 
 \caption{\method{}}
 \label{alg:graphweaver}
\end{algorithm}

%% file: sections/05-deployment.tex
\section{Deployment and Impact}\label{sec:deployment}
\textbf{Deployment.} 
\method{} has been deployed worldwide, serving hundreds of thousands of Microsoft Defender XDR customers over the last few months. 
The deployment infrastructure utilizes a Synapse-based PySpark cluster tailored to each geographical region (e.g., Europe, Asia). 
This setup includes: 
(a) an ADLS database that guarantees both accessibility and secure management of alert telemetry;
(b) an Azure Synapse backend that provides a robust framework for deployment and monitoring;
and (c) an XXL PySpark pool with 60 executors, each equipped with 64 CPU cores and 400GB of RAM.
To ensure system efficiency, Synapse is configured to autoscale the number of executors based on fluctuating job demand.
Additionally, \method{} runs every few minutes in each region to maintain optimal correlation coverage. 
To enhance system reliability and prevent potential data, Synapse automatically reruns any failed jobs.

\smallskip\noindent
\textbf{Impact.} 
Microsoft Defender XDR processes billions of correlations each month, with \method{} accounting for a significant portion of them. 
A key metric in the XDR domain is the singleton incident ratio---the percentage of incidents containing only a single alert. 
\method{} has successfully reduced the product's overall singleton incident rate by 7\%, translating into millions of investigation hours saved annually by security operation centers.
Customer feedback and detailed review of thousands of security incident correlations by our threat research team demonstrate that \method{} achieves a 99\% true correlation rate. 
Furthermore, the MST algorithm attains a 7.4x compression ratio on the number of correlations generated by our traditional correlation techniques, yielding substantial savings in terms of storage and downstream computational tasks.

%% file: sections/06-research-directions.tex
\section{Open Research Directions}\label{sec:research-directions}
The task of correlating vast quantities of security alerts presents several compelling research directions for future exploration. 
Each of these directions offers a significant opportunity to advance our understanding and capabilities in security alert management:

\vspace{1mm}\noindent
\textbf{Correlation verification at scale.}
With billions of correlations generated each month, it becomes infeasible to manually review more than a tiny fraction of them. 
Developing robust systems capable of verifying the accuracy of these correlations at scale is a key challenge. 
The integration of advanced techniques, such as large language models, could play a crucial role in achieving scalable verification. 
This will not only enhance the reliability of correlation processes but streamline security operations by reducing manual investigations.
The challenge here lies in creating a system that can efficiently process vast data sets while maintaining high precision.

\vspace{1mm}\noindent
\textbf{Enhanced decision making.}
Enhancing the decision-making process in correlations by incorporating security knowledge graphs~\cite{wang2017kgbiac, levshun2023survey} and sophisticated computational models, such as large language models~\cite{kuang2024knowledge}, 
is another promising research avenue. 
This research can significantly elevate the precision and relevance of correlations, especially for lower fidelity entity types like IP Ranges. 
The focus would be on developing systems that not only correlate data more accurately but also contextualize the significance of key correlations in real-world scenarios.

\vspace{1mm}\noindent
\textbf{Advanced gap detection.}
There is a critical need for the development of advanced correlation systems that are capable of identifying gaps in the connections between security alerts, where direct evidence is lacking. 
Techniques such as representation learning could be pivotal in mapping these gaps within a latent space, thereby uncovering hidden patterns and connections between seemingly unrelated alerts and incidents~\cite{levshun2023survey, chen2020identifying, chen2021graph, jin2020spectral}.

%% file: sections/07-conclusion.tex
\section{Conclusion}\label{sec:conclusion}
\method{} represents a groundbreaking shift in the landscape of large enterprise cybersecurity, marking the first time a cybersecurity company has openly discussed an industry-scale correlation framework in depth. 
Deployed globally to hundreds of thousands of customers as part of Microsoft Defender XDR, we are correlating billions of security alerts into coherent incidents.
\method{} introduces multiple novel ideas to the correlation space, including the use of a geo-distributed database and PySpark correlation engine, an MST algorithm for optimizing storage and downstream computational tasks, and a human-in-the-loop feedback system that enables continuous refinement of key correlation processes and parameters.
\method{} has decreased the product's singleton incident rate by 7\%, equating to millions of saved SOC investigation hours each year. 
This workload reduction is supported by a 99\% true correlation rate based on customer feedback and an extensive manual evaluation of thousands of correlations by our security team.
Our hope is that this research will not only advance correlation capabilities, but foster 
greater openness and collaboration within the community.

%% file: sections/08-acknowledgements.tex
\begin{acks}
We thank all of our colleagues who supported this research.
\end{acks}